\documentclass[journal,twocolumn]{IEEEtran}

\usepackage{amsmath,amssymb,amsfonts}
\usepackage{amsthm}
\usepackage{algorithm}
\usepackage{algorithmic}
\usepackage{bm}
\usepackage{booktabs}
\usepackage{cite}
\usepackage{graphicx}
\usepackage{mathtools}
\usepackage{url}

\newcommand{\E}{\mathbb{E}}
\newcommand{\R}{\mathbb{R}}
\newcommand{\jmathunit}{\mathrm{j}}
\newtheorem{theorem}{Theorem}
\newtheorem{proposition}{Proposition}
\newtheorem{corollary}{Corollary}

\begin{document}

\title{Why Is Cubic-Phase Airy Beamforming Sufficient for Blockage Recovery?}

\author{Yi~Wang and Linglong~Dai,~\IEEEmembership{Fellow,~IEEE}}

\maketitle

\begin{abstract}
Blockage is a critical challenge for near-field communications, where reliable
transmission depends heavily on the line-of-sight (LoS) path and can suffer
severe power degradation when that path is obstructed. Near-field Airy beams
offer a promising solution for blockage mitigation by forming curved
trajectories that guide energy around obstacles, and can be practically
generated with phased arrays by imposing a cubic source phase. However,
trajectory-based interpretations explain how Airy beams propagate, but not why
cubic-phase Airy beamforming is sufficient for blockage recovery or how much
received-power gain the cubic term itself contributes. To answer these
questions, we identify the blockage-induced phase mismatch relative to
conventional near-field focusing and quantify how successive phase orders
compensate it. The resulting analysis reveals that the linear and quadratic
degrees of freedom, originally used to compensate the free-space geometric
phase, can be reoptimized under blockage to provide resteering and refocusing,
respectively. The quadratic term can compensate the dominant quadratic component
of the additional mismatch, while the Airy cubic provides the first
independent correction to the remaining non-quadratic mismatch. Simulations
show that linear and quadratic compensation recover most of the available
gain. The Airy cubic adds only $0.1433$ dB on average, yet enables the
cubic-phase family to attain $99.77\%$ of the phase-only upper bound.
Residual-phase analysis further determines when the remaining higher-order
components are negligible within a prescribed received-power tolerance. These
results explain why cubic-phase Airy beamforming is sufficient: lower-order
phase terms provide most of the recovery, while the cubic term closes nearly
all of the remaining gap.
\end{abstract}

\begin{IEEEkeywords}
Airy beam, blockage avoidance, near-field communications, phase compensation.
\end{IEEEkeywords}

\section{Introduction}

High-frequency wireless communications, particularly in the millimeter-wave
(mmWave) and terahertz (THz) bands, offer abundant bandwidth. However, weak
non-line-of-sight (NLoS) propagation at these frequencies makes reliable
transmission highly dependent on the line-of-sight (LoS)
path~\cite{2021THz,8732419}. Such dependence makes LoS blockage a critical
challenge, as blockage can severely degrade received power and communication
reliability~\cite{shurakov2023empirical}. Near-field Airy beams with curved
trajectories can guide energy around obstacles, providing a promising approach
to blockage mitigation~\cite{overview,airyfinite}.

\subsection{Prior Works}

Airy beams were originally predicted as nonspreading wave packets in quantum
physics, and were later demonstrated experimentally in
optics~\cite{airy1,airyfinite,airy5}. The unique propagation properties of Airy
beams, particularly their curved trajectories, have been extensively
investigated in optical systems~\cite{overview,selfhealing1,selfhealing2}.

To leverage their curved trajectories to bypass obstacles, recent studies have applied Airy beams to blockage mitigation in wireless communications. 
Theoretical and numerical studies have
shown that Airy beams preserve received power under blockage more effectively
than conventional beams~\cite{hanchong,songlingyang,wang2026multiairy}. To realize this capability in practical communication systems,
phased-array methods have been developed to generate Airy beams by imposing a
cubic phase profile across the transmit aperture~\cite{hanchong2}. Hardware experiments have further
validated the feasibility of Airy-beam transmission in wireless
systems~\cite{lee,curving}.

To enable systematic trajectory design of Airy beams, closed-form analysis further showed
how a set of Airy control parameters can tune the steering, focusing, and
bending of the generated trajectory~\cite{hanchong2}. Different parameter settings therefore generate a continuum of candidate
trajectories, many of which can bypass the obstacle and reach the receiver.
The design task is thus not merely to find a feasible trajectory, but to identify the one that maximizes the received power. To reduce the resulting search overhead, beam-training methods organize candidate Airy configurations through structured or hierarchical codebooks and progressively narrow the trajectory space~\cite{hanchong,zhao2026efficienttraining}.
Alternatively, physics-guided prediction methods learn the mapping from
blockage geometry directly to the corresponding Airy control parameters,
thereby avoiding an explicit search over candidate trajectories
\cite{wang2026airy}.

However, existing trajectory-based interpretations do not quantitatively
explain why cubic-phase Airy beamforming is sufficient for blockage recovery,
or how much additional received-power gain the cubic term actually provides
beyond conventional focusing schemes. This gap calls for a phase-compensation
perspective that directly isolates and quantifies the role of each phase
order.

\subsection{Our Contributions}

To resolve the unclear phase-order mechanism underlying Airy blockage
recovery, we unveil its physical essence from a phase-compensation
perspective. Specifically, we identify the additional phase mismatch
introduced by blockage relative to conventional near-field focusing, and
quantify how successive phase orders compensate this mismatch and contribute
to received-power recovery. The contributions are summarized as follows:

\begin{itemize}

\item We identify the blockage-induced phase mismatch beyond conventional
near-field focusing. By comparing the phase-matched solution under blockage
with the conventional free-space focusing phase, the additional phase
correction required for blockage recovery is isolated. This phase-matched
solution also provides the phase-only upper bound for quantifying how much of
the required correction can be captured by each realizable compensation
order.

\item Based on the blockage-induced phase mismatch, we reveal that the
linear and quadratic phase degrees of freedom of conventional focusing can be
reused for blockage recovery. These terms originally compensate the free-space geometric phase, while under
blockage their coefficients can be reoptimized to provide resteering and
refocusing compensation, respectively.
The quadratic adjustment compensates the dominant quadratic component of the
blockage-induced phase mismatch, while the Airy cubic phase provides the first
independent correction to the remaining non-quadratic mismatch.

\item Building on the identified roles of different phase orders, we
quantitatively verify the sufficiency of the Airy cubic for blockage recovery.
We characterize the residual phase mismatch after cubic compensation and
determine when the remaining higher-order components are negligible within a
prescribed received-power tolerance. We further confirm that the required
cubic compensation can be realized within the standard Airy control family
through a received-contribution-weighted phase projection.

\item Numerical results validate the identified phase-order mechanism and the
sufficiency of the Airy cubic. Linear and quadratic compensation recover
$3.223$ dB of the $3.376$-dB mean phase-only gain, leaving only a
$0.1532$-dB residual gap. The Airy cubic contributes a further $0.1433$ dB
and reduces this gap to $0.00991$ dB, reaching $99.77\%$ of the phase-only
upper bound. These results confirm that the lower-order phase degrees of
freedom provide most of the blockage-recovery gain, while extending the phase
to cubic order is sufficient to recover nearly all of the remaining
available gain.

\end{itemize}

\subsection{Organization and Notation}

\subsubsection{Organization}

Section~\ref{sec:system} introduces the near-field system, the single-edge
blockage model, and the hardware-compatible Airy phase basis.
Section~\ref{sec:ideal_phase} derives the blockage-aware effective channel and
the resulting blockage-induced phase mismatch.
Section~\ref{sec:analytical_design} develops the phase-compensation hierarchy
and its weighted Airy realization, and interprets the cubic term as the first
non-quadratic correction.
Section~\ref{sec:implementation} presents the analytical realization of the
phase-compensation framework and its computational complexity.
Section~\ref{sec:results} provides the numerical evaluation, followed by the
conclusion. The appendixes collect the local edge and geometry extensions,
projection details, and supporting scope and numerical checks.

\subsubsection{Notation}
Bold lowercase and uppercase letters denote vectors and matrices, respectively.
The superscripts $(\cdot)^{\mathsf T}$ and $(\cdot)^*$ denote transpose and
complex conjugation, respectively, and $\|\cdot\|_2$ denotes the Euclidean
norm. The indicator function is denoted by $\mathbf 1_{\mathcal X}(\cdot)$.
The imaginary unit is $\jmathunit=\sqrt{-1}$, $\E[\cdot]$ denotes expectation,
and $\mathcal O(\cdot)$ denotes complexity order.

\section{System Model and Airy Phase Representation}
\label{sec:system}

This section establishes the blocked near-field system model and the
structured phase representation used throughout the paper. It first
defines the transceiver geometry and single-edge blockage model, then
introduces the Gaussian-cubic phase family as a hardware-compatible
polynomial phase basis for subsequent phase-compensation analysis.

\subsection{Near-Field Geometry}

We consider a narrowband near-field communication system in the
$x$-$z$ plane. An $N$-element uniform linear array (ULA) is centered at the
origin and placed along the $x$-axis. The coordinate of its $n$-th element is
\begin{equation}
    x_n=\left(n-\frac{N+1}{2}\right)d_{\rm ant},\qquad n=1,\ldots,N,
    \label{eq:element_position}
\end{equation}
where $d_{\rm ant}\leq\lambda/2$ is the antenna-element spacing and $\lambda$
is the wavelength. The aperture width is $D=(N-1)d_{\rm ant}$, and its half-width is
$a=D/2$. Throughout this paper, a spatial point is represented as $(z,x)$,
where $z$ denotes the propagation distance and $x$ denotes the transverse
coordinate.

The receiver is centered at $(z_r,x_r)$ and is represented by an effective
one-dimensional aperture
\begin{equation}
    \mathcal R=
    \left[x_r-\frac{W_r}{2},x_r+\frac{W_r}{2}\right],
    \label{eq:receiver_window}
\end{equation}
where $W_r$ is the receiver-window width. The radiative near-field model
retains the distance-dependent phase variation across the finite aperture
rather than replacing it with a far-field steering vector.

\subsection{Single-Edge Blockage Model}

An opaque obstacle is located between the transmitter and receiver. Its
locally dominant boundary is represented by the edge $(z_o,x_e)$, where
$0<z_o<z_r$. The visible half-plane at $z=z_o$ is
\begin{equation}
    \mathcal V_s=\{x:s(x-x_e)>0\},\qquad s\in\{-1,+1\},
    \label{eq:visible_half_plane}
\end{equation}
where $s$ identifies the visible side. This model isolates the dominant
edge-diffraction contribution and provides an analytically tractable
representation of the blockage-induced wavefront distortion. The transmitter
is assumed to know
\begin{equation}
    \mathcal I=(z_o,x_e,s,z_r,x_r),
    \label{eq:geometry_information}
\end{equation}
from environment sensing or a higher-layer map, while complete post-blockage
channel state information is not assumed.
The resulting single-edge geometry is illustrated in
Fig.~\ref{fig:system_single_edge}.

\begin{figure}[!t]
    \centering
    \includegraphics[width=\columnwidth]{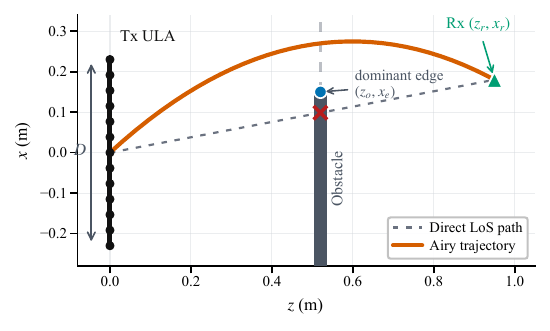}
    \caption{Single-edge blocked near-field link. A phase-only transmit (Tx)
    ULA of aperture $D$ serves the receiver (Rx) at $(z_r,x_r)$. The obstacle terminates at
    the dominant edge $(z_o,x_e)$ and blocks the direct propagation path. The
    orange curve illustrates the resulting propagated field trajectory.}
    \label{fig:system_single_edge}
\end{figure}

Under the scalar Fresnel approximation, the one-dimensional free-space kernel
over distance $z$ is~\cite{optics}
\begin{equation}
    K_z(x_2,x_1)=
    \frac{e^{\jmathunit kz}}{\sqrt{\jmathunit\lambda z}}
    \exp\left[
    \frac{\jmathunit\pi}{\lambda z}(x_2-x_1)^2
    \right],
    \label{eq:fresnel_kernel}
\end{equation}
where $k=2\pi/\lambda$. For a continuous aperture field $u_0(x_0)$
supported on $[-a,a]$, the blocked receiver-plane field is
\begin{equation}
\begin{aligned}
    E_{\rm b}(x)=
    \int_{-a}^{a}u_0(x_0)
    \int_{\mathcal V_s}
    K_{z_r-z_o}(x,\xi)
    K_{z_o}(\xi,x_0)
    \,\mathrm d\xi\,\mathrm dx_0.
\end{aligned}
\label{eq:two_step_blocked_field}
\end{equation}

The above propagation model defines the physical channel environment in which
the source phase is designed. We next introduce the structured Airy phase
basis used to represent the required compensation.

\subsection{Hardware-Compatible Airy Phase Basis}

To investigate how structured phase profiles compensate the blockage-induced
wavefront distortion, we consider the Gaussian-cubic Airy realization as a
hardware-compatible polynomial phase basis. All considered beams share the
same prescribed Gaussian aperture amplitude, and only the source phase is
varied. The continuous source field is
\begin{equation}
    u_0(x_0;\bm p)
    =
    A_G(x_0)e^{\jmathunit\phi_{\bm p}(x_0)},
    \qquad
    \bm p=(B,F,\theta),
    \label{eq:source_field}
\end{equation}
where
\begin{equation}
    A_G(x_0)=
    c_G\exp(-x_0^2/\omega_0^2)
    \mathbf 1_{[-a,a]}(x_0),
    \label{eq:gaussian_aperture}
\end{equation}
is a finite Gaussian envelope, where $\omega_0$ is the Gaussian width and
$c_G$ normalizes the transmit power.
The control triplet $\bm p$ is restricted to the physically feasible Airy
control set $\mathcal P_{\rm A}$.

Fixing the aperture amplitude across all methods isolates the contribution of
source phase. Amplitude control, such as visibility-based antenna
deactivation or independent amplitude optimization, changes the feasible
transmit architecture and is therefore outside the considered phase-only
framework. Within this architecture, the phase-matched reference
derived in Section~\ref{sec:ideal_phase} defines the upper bound attainable
with the prescribed aperture amplitude.

The Gaussian-cubic phase profile is
\begin{equation}
    \phi_{\bm p}(x_0)
    =
    \frac{(2\pi B)^3}{3}x_0^3
    -\frac{\pi}{\lambda F}x_0^2
    +\frac{2\pi\sin\theta}{\lambda}x_0.
    \label{eq:gaussian_cubic_phase}
\end{equation}

Here, $(B,F,\theta)$ are the conventional Airy beam control parameters that
parameterize the cubic phase basis through bending, focusing, and steering
degrees of freedom.
Conventional Airy generation interprets these parameters through a
geometry-defined construction, where $F$ and $\theta$ define the generation
configuration and $B$ controls the cubic-phase-induced trajectory curvature.
Fig.~\ref{fig:conventional_airy_generation} illustrates this design
paradigm.

\begin{figure}[!t]
    \centering
    \includegraphics[width=\columnwidth]{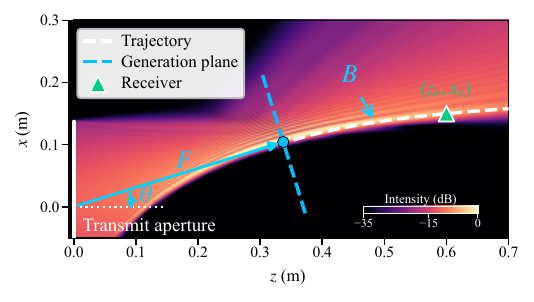}
    \caption{Conventional geometry-defined Airy generation: $F$ and $\theta$
    determine the generation configuration, while $B$ controls the subsequent
    trajectory curvature. No blockage-aware channel phase is considered.}
    \label{fig:conventional_airy_generation}
\end{figure}

In this paper, the triplet $(B,F,\theta)$ is retained only as a
hardware-compatible parameterization of the polynomial phase basis used to
represent the required compensation. Any resulting generation plane or
trajectory is interpreted as a propagated-field consequence rather than the
design objective. With the normalized aperture coordinate
$t=x_0/a$, \eqref{eq:gaussian_cubic_phase} becomes
\begin{equation}
    \phi_{\bm a}(t)=a_1t+a_2t^2+a_3t^3,
    \label{eq:normalized_phase}
\end{equation}
where $\bm a=(a_1,a_2,a_3)$ and
\begin{equation}
\begin{aligned}
	a_1&=\frac{2\pi a\sin\theta}{\lambda},&
	a_2&=-\frac{\pi a^2}{\lambda F},&
	a_3&=\frac{(2\pi B)^3a^3}{3}.
	\label{eq:control_coefficient_map}
\end{aligned}
\end{equation}
The map is one-to-one on the physical branch $F>0$ and
$|\sin\theta|<1$, with inverse
\begin{equation}
\begin{aligned}
	B&=\frac{\sqrt[3]{3a_3}}{2\pi a},&
	F&=-\frac{\pi a^2}{\lambda a_2},&
	\theta&=\arcsin\left(\frac{\lambda a_1}{2\pi a}\right),
	\label{eq:coefficient_control_inverse}
\end{aligned}
\end{equation}
where $\sqrt[3]{\cdot}$ denotes the real cube root.
For a practical ULA implementation, the continuous phase profile is sampled
on the antenna elements as
\begin{equation}
	\bm w(\bm p)
	=\frac{
	[A_G(x_1)e^{\jmathunit\phi_{\bm p}(x_1)},\ldots,
	A_G(x_N)e^{\jmathunit\phi_{\bm p}(x_N)}]^{\mathsf T}
	}{
	\|[A_G(x_1)e^{\jmathunit\phi_{\bm p}(x_1)},\ldots,
	A_G(x_N)e^{\jmathunit\phi_{\bm p}(x_N)}]\|_2
	}.
	\label{eq:discrete_beamforming_vector}
\end{equation}

This structured Airy phase family will be used in the following sections to
represent the phase compensation required by the blocked channel.

\section{Blockage-Induced Phase Mismatch}
\label{sec:ideal_phase}

The structured Airy phase basis in \eqref{eq:normalized_phase} specifies the
available source-phase space but does not reveal how blockage modifies the
required wavefront. This section extracts the missing phase requirement from
the blocked propagation model. It first combines the two Fresnel propagation
segments into a blockage-aware effective channel and then derives the
additional phase mismatch beyond the conventional free-space focusing phase.

\subsection{Blockage-Aware Effective Channel}

Each aperture point reaches the receiver through two Fresnel propagation
segments coupled by the obstacle edge. To expose how blockage modifies the
required source phase, we combine these operations into a single
blockage-aware effective channel coefficient. Let
$L=z_r-z_o$ and $L_{\rm eff}=z_oL/z_r$. Completing the square in the two
propagation kernels gives
\begin{equation}
    \frac{(\xi-x_0)^2}{z_o}+\frac{(x-\xi)^2}{L}
    =
    \frac{(x-x_0)^2}{z_r}
    +
    \frac{[\xi-\bar x_o(x_0;x)]^2}{L_{\rm eff}},
    \label{eq:complete_square}
\end{equation}
where
\begin{equation}
    \bar x_o(x_0;x)
    =
    \frac{Lx_0+z_ox}{z_r}.
    \label{eq:composed_center}
\end{equation}

Let
$\alpha=\sqrt{\pi/(\lambda L_{\rm eff})}$.
Using the principal Fresnel branch, the half-line integral becomes
\begin{equation}
\begin{aligned}
&\int_{\mathcal V_s}
\exp\left[
\frac{\jmathunit\pi}{\lambda L_{\rm eff}}
(\xi-\bar x_o)^2
\right]\mathrm d\xi\\
&\quad=
\frac{\sqrt{\lambda L_{\rm eff}}e^{\jmathunit\pi/4}}{2}
\operatorname{erfc}
\left[
s e^{-\jmathunit\pi/4}
\alpha(x_e-\bar x_o)
\right].
\end{aligned}
\label{eq:fresnel_half_line}
\end{equation}
Here, $\operatorname{erfc}(\cdot)$ denotes the complementary error function.

For $s=+1$ and $s=-1$, \eqref{eq:fresnel_half_line} corresponds to the
integration regions $[x_e,+\infty)$ and $(-\infty,x_e]$, respectively. The
common full-line factor combines with the Fresnel kernels to yield the direct
kernel $K_{z_r}$. Therefore,
\begin{equation}
    \int_{\mathcal V_s}
    K_L(x,\xi)K_{z_o}(\xi,x_0)\mathrm d\xi
    =
    K_{z_r}(x,x_0)T_{\rm e}(x_0;x),
    \label{eq:edge_kernel_composition}
\end{equation}
where
\begin{equation}
    T_{\rm e}(x_0;x)
    =
    \frac{1}{2}
    \operatorname{erfc}
    \left[
    s e^{-\jmathunit\pi/4}
    \sqrt{\frac{\pi}{\lambda L_{\rm eff}}}
    (x_e-\bar x_o(x_0;x))
    \right].
    \label{eq:edge_transfer}
\end{equation}

The term $T_{\rm e}(x_0;x)$ represents the edge-induced channel distortion,
including both amplitude attenuation and phase variation. Hence, the blocked
field can be reduced exactly within the scalar Fresnel model to
\begin{equation}
    E_{\rm b}(x;\bm a)
    =
    \int_{-a}^{a}
    A_G(x_0)
    h_{\rm e}(x_0;x)
    e^{\jmathunit\phi_{\bm a}(x_0/a)}
    \,\mathrm dx_0,
    \label{eq:single_integral_field}
\end{equation}
where
\begin{equation}
    h_{\rm e}(x_0;x)
    =
    K_{z_r}(x,x_0)T_{\rm e}(x_0;x).
    \label{eq:effective_aperture_channel}
\end{equation}

Therefore, $h_{\rm e}(x_0;x)$ describes the blockage-aware effective channel
from the transmit aperture point $x_0$ to the receiver-plane point $x$. The
free-space kernel provides the conventional receiver-directed focusing phase,
whereas the edge-diffraction response introduces the blockage-induced
amplitude attenuation and additional phase distortion.

\subsection{Phase-Matched Reference and Phase Mismatch}

For the phase-compensation analysis, we consider the point-receiver
received-power objective
\begin{equation}
    P_{\rm pt}(u_0;\mathcal I)
    =
    |E_{\rm b}(x_r)|^2 ,
    \label{eq:point_received_power}
\end{equation}
where the physical meaning of phase compensation is most explicit. The
finite receiver-window objective is considered in the numerical evaluation.

For a point receiver and a fixed aperture amplitude, the maximum achievable
received power is obtained when all surviving aperture contributions are
phase aligned. Using \eqref{eq:effective_aperture_channel} at $x_r$, write
\begin{equation}
    h_{\rm e}(x_0;x_r)
    =
    |h_{\rm e}(x_0;x_r)|
    e^{\jmathunit\psi_{\rm h}(x_0)},
    \label{eq:effective_channel_polar}
\end{equation}
and define the received-contribution weight
\begin{equation}
    W(x_0)
    =
    A_G(x_0)|h_{\rm e}(x_0;x_r)|.
    \label{eq:continuous_weight}
\end{equation}
The triangle inequality gives
\begin{equation}
    |E_{\rm b}(x_r)|
    \leq
    \int_{-a}^{a}W(x_0)\,\mathrm dx_0 .
    \label{eq:triangle_bound}
\end{equation}
Equality is achieved when the contributions from all aperture points are
coherently aligned. The resulting phase-matched reference is used only as an
analytical reference to expose the required compensation, rather than as a
practical beamforming design. It is given by
\begin{equation}
    \boxed{
    \phi_{\rm ideal}(x_0)
    =
    -\arg h_{\rm e}(x_0;x_r)+\phi_0 ,
    }
    \label{eq:ideal_phase}
\end{equation}
where $\phi_0$ is an arbitrary common phase offset.
Substituting
\eqref{eq:effective_aperture_channel} into
\eqref{eq:ideal_phase} yields
\begin{equation}
    \phi_{\rm ideal}
    =
    -\arg K_{z_r}
    -
    \arg T_{\rm e}
    +
    \phi_0 .
    \label{eq:phase_decomposition}
\end{equation}
Using the normalized aperture coordinate $t=x_0/a$, the conventional
free-space focusing phase is defined as
\begin{equation}
    \phi_{\rm F}(t)
    =
    -\arg K_{z_r}(x_r,at)+\phi_0 .
    \label{eq:free_space_focused_phase}
\end{equation}
The additional phase requirement introduced by blockage is therefore
\begin{equation}
    \boxed{
    \Delta\phi_{\rm e}(t)
    =
    \phi_{\rm ideal}(at)-\phi_{\rm F}(t)
    =
    -\arg T_{\rm e}(at;x_r).
    }
    \label{eq:blockage_phase_mismatch}
\end{equation}
All phases are evaluated on continuous branches over the illuminated aperture,
and common phase offsets do not affect the received power. Equation
\eqref{eq:blockage_phase_mismatch} is the central variable of this paper:
conventional focusing compensates the free-space geometric phase, whereas
$\Delta\phi_{\rm e}$ captures the additional phase correction required beyond
conventional focusing by the blockage-aware channel. It vanishes in the
unobstructed limit and generally varies nonlinearly across the aperture under
blockage.

Fig.~\ref{fig:phase_compensation_interpretation} summarizes the resulting
phase-compensation view after the effective channel has been formulated and
the phase mismatch has been extracted.

\begin{figure*}[!t]
    \centering
    \includegraphics[width=0.94\textwidth]{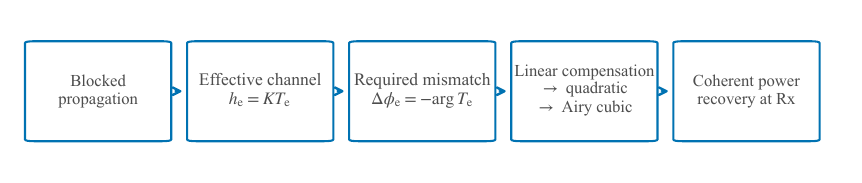}
    \caption{Phase-compensation interpretation of Airy recovery under edge
    blockage. Rather than prescribing a generation plane or trajectory,
    blocked propagation determines the effective channel $h_{\rm e}$ and the
    additional phase mismatch $\Delta\phi_{\rm e}$. Successive linear, quadratic,
    and cubic phase components approximate the required phase compensation and
    enable coherent received-power recovery. The trajectory follows afterward
    as a propagated-field consequence.}
    \label{fig:phase_compensation_interpretation}
\end{figure*}

Equation~\eqref{eq:ideal_phase} provides the phase-only upper bound for
evaluating realizable compensation stages. It defines the blockage-induced
phase requirement independently of the Airy parameterization. The extracted
mismatch $\Delta\phi_{\rm e}$ therefore provides the basis for the
phase-compensation hierarchy developed in the following section.

Finite-window perturbation analysis and its numerical validation are provided
in Appendices~\ref{app:theoretical_details}
and~\ref{app:numerical_validation}, respectively.

\section{Phase-Compensation Hierarchy and Airy Phase Interpretation}
\label{sec:analytical_design}

The central object is the blockage-induced mismatch
$\Delta\phi_{\rm e}$ defined in \eqref{eq:blockage_phase_mismatch}.
This section reveals how the required phase compensation can be progressively
approximated by linear, quadratic, and cubic components. The Gaussian-cubic
Airy family is then interpreted as a hardware-compatible realization of this
compensation hierarchy.
Fig.~\ref{fig:phase_hierarchy_schematic} previews the resulting hierarchy.

\begin{figure}[!t]
	\centering
	\includegraphics[width=\columnwidth]{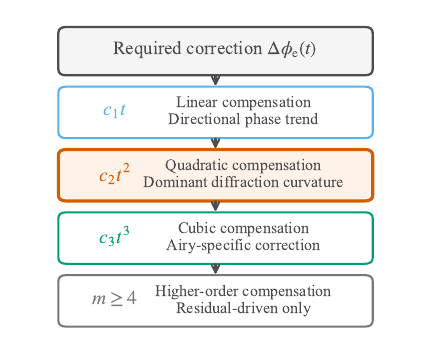}
	\caption{Phase-compensation hierarchy for blockage recovery. Linear,
	quadratic, and cubic phase components progressively approximate the
	blockage-induced mismatch. The resulting Airy trajectory emerges only as a
	propagated-field consequence of the compensated source phase.}
	\label{fig:phase_hierarchy_schematic}
\end{figure}

\subsection{Phase-Compensation Hierarchy}
\label{subsec:nested_compensation}

Because conventional near-field focusing already contains receiver-directed
linear and quadratic terms, the order below refers to the \emph{additional}
edge-induced phase compensation beyond that free-space focusing phase. Globally over the
aperture, we represent this compensation by the hierarchy
\begin{equation}
	\Delta\phi_{\rm e}(t)
	\approx c_0+c_1t+c_2t^2+c_3t^3+\cdots.
	\label{eq:blockage_phase_hierarchy}
\end{equation}
Here, $c_0$ is an irrelevant common phase, and $c_\ell$ is the
order-$\ell$ blockage-compensation coefficient for $\ell=1,2,3$.
This approximation describes the global aperture-phase compensation required by
the blocked channel. The linear term provides additional steering
compensation, while the quadratic degree of freedom enables refocusing
compensation by adjusting the existing focusing curvature. Under blockage,
this quadratic degree of freedom becomes the dominant mechanism for
compensating the curvature mismatch. The Airy cubic component is the first
non-quadratic correction.

Fig.~\ref{fig:phase_mismatch_extraction} extracts the blockage-induced mismatch
approximated by this hierarchy. Comparing the phase-matched reference with
conventional free-space focusing isolates the additional requirement
$\Delta\phi_{\rm e}$ introduced by blockage.
\begin{figure}[!t]
	\centering
	\includegraphics[width=\columnwidth]{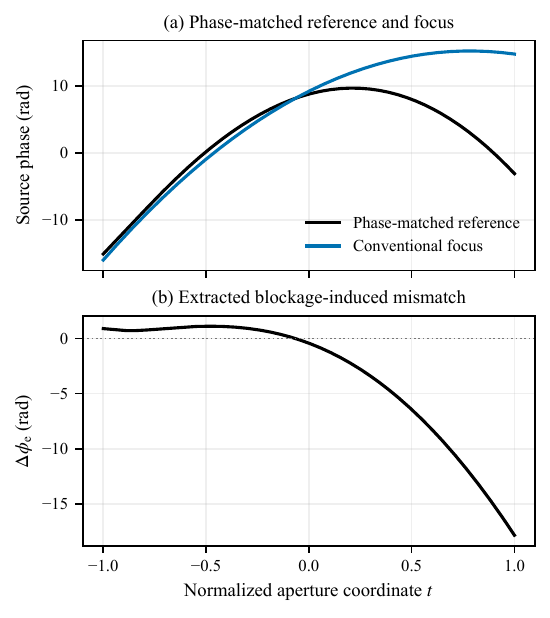}
	\caption{Extraction of the blockage-induced phase mismatch.
	(a) The phase-matched reference differs from conventional free-space
	focusing under blockage. (b) Removing the free-space focusing phase isolates
	$\Delta\phi_{\rm e}$. Common phase offsets are removed because they do not
	affect received power.}
	\label{fig:phase_mismatch_extraction}
\end{figure}

The reference phase $\phi_{\rm F}(t)$ is defined in
\eqref{eq:free_space_focused_phase}; it already contains the receiver-directed
linear term and near-field quadratic curvature. Let $P_{\rm F}$ denote its
exact received power. To separate the role of each added correction, define
\begin{equation}
\begin{aligned}
	\mathcal S_1
	&=\{\phi_{\rm F}+c_0+c_1t\},\\
	\mathcal S_2
	&=\{\phi_{\rm F}+c_0+c_1t+c_2t^2\},\\
	\mathcal S_3
	&=\{\phi_{\rm F}+c_0+c_1t+c_2t^2+c_3t^3\}.
	\label{eq:nested_compensation_spaces}
\end{aligned}
\end{equation}
These spaces are nested:
\begin{equation}
	\{\phi_{\rm F}\}\subset\mathcal S_1
	\subset\mathcal S_2\subset\mathcal S_3.
	\label{eq:nested_phase_inclusion}
\end{equation}
Here, $\mathcal S_1$ adds additional linear compensation,
$\mathcal S_2$ adds refocusing compensation, and $\mathcal S_3$ introduces the
independent Airy cubic correction. Although the quadratic correction is
implemented by changing the focusing coefficient, its received-power gain
originates from compensating the blockage-induced curvature mismatch rather
than merely relocating the focal point.
Mathematically, the cubic term is the first non-quadratic phase degree of
freedom. The
physically feasible subset of $\mathcal S_3$ is precisely the Gaussian-cubic
Airy phase family in \eqref{eq:normalized_phase}.
The steering-only baseline used later is distinct from $\mathcal S_1$: it
omits the free-space focusing curvature, whereas every space
$\mathcal S_m$ starts from the conventional free-space focusing phase.

For a fair gain decomposition, let $P_m^\star$ be the maximum exact received power
over $\mathcal S_m$ within the adopted control domain. From
\eqref{eq:nested_phase_inclusion},
\begin{equation}
	P_{\rm F}\leq P_1^\star\leq P_2^\star\leq P_3^\star.
	\label{eq:nested_power_order}
\end{equation}
These best-power references quantify the power made available at each
beamforming stage. They serve only as evaluation references rather than the
proposed online procedure. The one-shot Airy-stage approximation is derived
next.

\subsection{Weighted Phase Projection and Airy Realization}

The compensation hierarchy identifies the phase compensation required by the
blocked channel. We next project it onto the Gaussian-cubic Airy phase basis,
rather than searching over Airy configurations through iterative
received-power optimization. Because
aperture points with larger $W$ in \eqref{eq:continuous_weight} contribute more strongly to the
received field, the projection weights their phase errors according to their
received contribution. To construct this contribution-weighted projection,
let
\begin{equation}
	\bm X(t)=[1,t,t^2,t^3]^{\mathsf T},\qquad
	\psi_{\rm req}(t)=-\psi_{\rm h}(at),
	\label{eq:projection_basis}
\end{equation}
and define
\begin{equation}
	Z_W=\int_{-1}^{1}W(at)\,\mathrm dt,\qquad
	\E_W[f]=\frac{1}{Z_W}\int_{-1}^{1}W(at)f(t)\,\mathrm dt.
\end{equation}
Because the free-space focusing phase is quadratic, fitting
$\psi_{\rm req}$ through degree two or higher is equivalent to starting from
$\phi_{\rm F}$ and fitting the blockage-induced mismatch
$\Delta\phi_{\rm e}$. Thus, the fitted linear and quadratic coefficients
combine the conventional focus with the required linear and quadratic
compensation, while extending the fit to cubic order adds the first
non-quadratic correction.
The constant component in
$\bm q=[c_0,a_1,a_2,a_3]^{\mathsf T}$ absorbs the irrelevant common phase.

When the fitted phase remains close enough to the phase-matched reference for the aperture
contributions to combine coherently, received-power loss is governed to
second order by the weighted phase-error variance. The following theorem
makes this relation explicit.

\begin{theorem}[Small-mismatch phase fitting]
\label{thm:phase_projection}
Let
\begin{equation}
	\delta(t)
	=\bm q^{\mathsf T}\bm X(t)-\psi_{\rm req}(t)
	-\E_W[\bm q^{\mathsf T}\bm X-\psi_{\rm req}],
	\label{eq:centered_phase_residual}
\end{equation}
be the centered phase mismatch. If
$\|\delta\|_{\infty}\leq\varepsilon_{\phi}$ for some
$\varepsilon_{\phi}>0$, then
\begin{equation}
	\frac{|E_{\rm b}(x_r)|^2}{(aZ_W)^2}
	=1-\E_W[\delta^2]
	+\mathcal O\bigl(\E_W|\delta|^3\bigr).
	\label{eq:coherent_power_expansion}
\end{equation}
Consequently, the second-order power approximation in the cubic phase family
is maximized by
\begin{equation}
	\boxed{
	\widehat{\bm q}
	=\left(\int_{-1}^{1}W(at)\bm X\bm X^{\mathsf T}\,\mathrm dt\right)^{-1}
	\int_{-1}^{1}W(at)\bm X\psi_{\rm req}\,\mathrm dt,
	}
	\label{eq:weighted_phase_projection}
\end{equation}
provided that the weighted Gram matrix is nonsingular.
\end{theorem}

This result converts the nonlinear coherent-combining objective into a
fixed-dimensional weighted phase projection.
The result follows by expanding the weighted coherent sum around zero phase
mismatch. Appendix~\ref{app:theoretical_details} provides the complete remainder
bound and normal-equation derivation.

The fitted coefficients map one-to-one to $(B,F,\theta)$ through
\eqref{eq:coefficient_control_inverse}, providing the hardware-compatible Airy
realization of the identified phase compensation.

\subsection{Airy Cubic as the First Non-Quadratic Correction}

The phase-compensation perspective reveals that the Airy cubic is not the
origin of the dominant blockage-recovery gain; instead, it compensates the
residual non-quadratic mismatch left after lower-order compensation. The
quadratic component has already removed the dominant blockage-induced
curvature mismatch.
The raw basis $t^3$ is not an independent Airy-specific correction under the
weighted fit because it also contains components that can be represented by
constant, linear, and quadratic phases. To isolate only the new action beyond
linear and quadratic compensation, we remove these lower-order components from
$t^3$. The purpose of this projection is to isolate the cubic phase action that
cannot be reproduced by steering and refocusing. The free-space phase
$-\arg K_{z_r}$ belongs to
\begin{equation}
	\mathcal V_2=\operatorname{span}\{1,t,t^2\}.
\end{equation}
Let $\Pi_2^W$ denote weighted projection onto $\mathcal V_2$. Define
\begin{equation}
	r_3(t)=t^3-\Pi_2^W(t^3),
	\label{eq:orthogonal_cubic_basis}
\end{equation}
as the independent Airy cubic component after removing all weighted
lower-order compensation components. Mathematically, this is an orthogonal cubic
basis under the weighted inner product
\begin{equation}
	\langle f,g\rangle_W
	=\int_{-1}^{1}W(at)f(t)g(t)\,\mathrm dt.
	\label{eq:weighted_inner_product}
\end{equation}
The induced norm is $\|f\|_W^2=\langle f,f\rangle_W$. Define
\begin{equation}
	\mathcal V_3=\mathcal V_2\oplus\operatorname{span}\{r_3\},
\end{equation}
and let $\Pi_3^W$ denote weighted projection onto $\mathcal V_3$. For
$m\in\{2,3\}$, define
\begin{equation}
	e_m=(I-\Pi_m^W)\psi_{\rm req},
\end{equation}
where $I$ is the identity operator. Let $V_m=\E_W[e_m^2]$ be the weighted
variance of the phase mismatch
remaining after stage $m$. Because the free-space phase already belongs to
$\mathcal V_2$, only the edge-induced phase determines the fitted Airy cubic
coefficient:
\begin{equation}
	\boxed{
	a_3^{\rm proj}
	=\frac{\langle r_3,e_2\rangle_W}
	{\langle r_3,r_3\rangle_W}
	=-\frac{\langle r_3,\arg T_{\rm e}\rangle_W}
	{\langle r_3,r_3\rangle_W}.
	}
	\label{eq:cubic_edge_projection}
\end{equation}
Here, $e_2$ is the residual phase mismatch after linear and quadratic
compensation. Adding the Airy cubic can remove only the part of $e_2$ captured
by the independent component $r_3$. The next proposition quantifies this
physical statement.
\begin{proposition}[Cubic phase-mismatch capture]
\label{prop:cubic_capture}
The phase mismatch after adding the best cubic correction satisfies
\begin{equation}
	e_3=e_2-
	\frac{\langle e_2,r_3\rangle_W}
	{\langle r_3,r_3\rangle_W}r_3,
\end{equation}
and hence
\begin{equation}
	\|e_2\|_W^2-\|e_3\|_W^2
	=\frac{\langle e_2,r_3\rangle_W^2}
	{\langle r_3,r_3\rangle_W}.
	\label{eq:cubic_residual_reduction}
\end{equation}
For $V_2>0$, the fraction of phase-mismatch variance left after quadratic
compensation and captured by the Airy cubic is therefore
\begin{equation}
	\boxed{
	\eta_3
	=\frac{V_2-V_3}{V_2}
	=\frac{\langle e_2,r_3\rangle_W^2}
	{\langle e_2,e_2\rangle_W\langle r_3,r_3\rangle_W},
	}
	\label{eq:cubic_capture_factor}
\end{equation}
Thus, $\eta_3$ directly quantifies how much of the residual mismatch after
refocusing can be captured by the independent Airy cubic component.
\end{proposition}

A complete proof is provided in Appendix~\ref{app:theoretical_details}.

Let $P_m^{\rm proj}$ denote the exact point-receiver power produced by the
stage-$m$ phase fit and let $P_{\rm ub}=(aZ_W)^2$. Applying
Theorem~\ref{thm:phase_projection} to $e_2$ and $e_3$ in the small-mismatch regime
gives
\begin{equation}
\begin{aligned}
	\frac{P_3^{\rm proj}-P_2^{\rm proj}}
	{P_{\rm ub}-P_2^{\rm proj}}
	&=\eta_3+
	\mathcal O\left(\frac{R_2+R_3}{V_2}\right),\\
	R_m&=\E_W[|e_m|^3].
\end{aligned}
	\label{eq:cubic_power_gap_capture}
\end{equation}
Equivalently, the leading-order cubic gain in decibels is
\begin{equation}
	\Delta G_3
	=\frac{10}{\ln 10}\eta_3V_2
	+\mathcal O(V_2^2+R_2+R_3).
	\label{eq:cubic_gain_factorization}
\end{equation}
Equations~\eqref{eq:cubic_capture_factor}--\eqref{eq:cubic_gain_factorization}
separate two communication roles: $V_2$ quantifies the residual phase mismatch after
linear and quadratic compensation, while $\eta_3$ measures how much of this
mismatch is captured by the Airy cubic. Therefore, a small cubic gain does not
indicate a weak cubic correction; it reflects that the dominant quadratic
stage leaves only a small residual phase mismatch.
The Airy cubic component is not the origin of the blockage-recovery gain; instead, it
compensates the lowest-order non-quadratic component of that residual. Within
the Gaussian-cubic Airy family, the cubic term therefore represents the first
Airy-specific phase action beyond steering and refocusing compensation.

Fig.~\ref{fig:hierarchical_phase_compensation} visualizes the successive
residual reduction for the same representative geometry.
\begin{figure}[!t]
	\centering
	\includegraphics[width=\columnwidth]{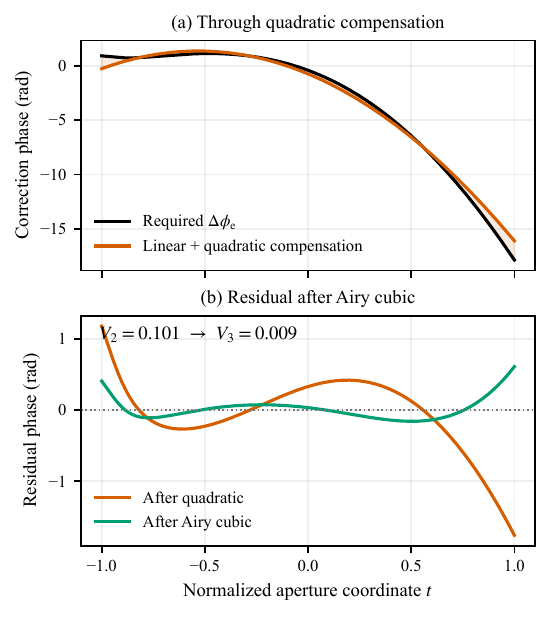}
	\caption{Hierarchical compensation of the blockage-induced phase mismatch.
	(a) Linear compensation removes the first-order variation, while quadratic
	compensation removes the dominant curvature variation of $\Delta\phi_{\rm e}$.
	(b) The Airy cubic component reduces the remaining non-quadratic
	residual from variance $V_2$ to $V_3$.}
	\label{fig:hierarchical_phase_compensation}
\end{figure}

\section{Analytical Airy Realization Framework}
\label{sec:implementation}

The phase hierarchy identifies the physical mechanism of blockage recovery.
This section converts the identified phase-compensation requirement into a
practical Gaussian-cubic Airy realization. Specifically, the required phase
compensation is extracted from the blocked channel, projected onto the
polynomial Airy phase basis, and mapped to realizable Airy control parameters.

\subsection{One-Shot Airy Realization}

The analytical realization combines Theorem~\ref{thm:phase_projection} with
the Airy coefficient mapping in \eqref{eq:coefficient_control_inverse}. It
provides a direct geometry-to-Airy mapping through three stages: constructing
the blockage-aware phase-matched target, projecting this target onto the
Gaussian-cubic Airy basis, and converting the resulting coefficients into
Airy controls.

\noindent\textbf{1) Blockage-aware phase-matched target:}
Given the sensed geometry $\mathcal I$ and antenna positions $\{x_n\}$, the
first stage constructs the discrete phase target required for coherent
recovery. The effective channel is evaluated according to
\eqref{eq:effective_aperture_channel}. At $t_n=x_n/a$, the free-space focusing
phase and blockage-induced phase mismatch are sampled from
\eqref{eq:free_space_focused_phase} and
\eqref{eq:blockage_phase_mismatch}, respectively. Their sum gives the
phase-matched target in \eqref{eq:ideal_phase}, up to an irrelevant common
phase offset. The samples are ordered by increasing $x_n$ and continuously
unwrapped before projection so that artificial $2\pi$ discontinuities are not
interpreted as high-order aperture-phase variation. The resulting unwrapped
phase sequence is the input to the projection stage.

\noindent\textbf{2) Contribution-weighted phase projection:}
The second stage maps the unwrapped phase target to the cubic polynomial
coefficients using the sampled projection in
Corollary~\ref{cor:discrete_projection}. This fixed-dimensional projection
replaces iterative Airy-parameter search. Each element is weighted by
$w_n=A_G(x_n)|h_{\rm e}(x_n;x_r)|$, i.e., its fixed Gaussian
amplitude multiplied by the magnitude of its effective channel. Phase errors
at elements that contribute more strongly to the received field are therefore
penalized more heavily than errors at weakly contributing elements. The
weighted Gram matrix and phase moments yield
$(c_0,a_1,a_2,a_3)$ through the fixed-dimensional system in
\eqref{eq:discrete_projection}, where $c_0$ is a common phase and the remaining
coefficients provide the linear, quadratic, and cubic corrections.

\noindent\textbf{3) Airy control conversion:}
The final stage converts the fitted phase coefficients into feasible Airy
controls and the corresponding beamforming vector. The three phase
coefficients are first checked against the feasible Airy control range. If the
unconstrained solution violates this range, the bounded
projection in \eqref{eq:constrained_phase_projection} is applied using the
same weighted moments. The feasible coefficients are then converted into
$(B,F,\theta)$ through \eqref{eq:coefficient_control_inverse}, and the sampled
Gaussian-cubic phase generates the practical ULA beamforming vector according
to \eqref{eq:discrete_beamforming_vector} under the fixed Gaussian aperture
envelope. Algorithm~\ref{alg:one_shot_airy} summarizes the complete
realization.

\begin{algorithm}[!t]
\caption{One-Shot Airy Phase Realization}
\label{alg:one_shot_airy}
\begin{algorithmic}[1]
\REQUIRE Geometry $\mathcal I$, $\lambda$, $a$, $\omega_0$, and array
positions $\{x_n\}_{n=1}^{N}$
\ENSURE Airy controls $\bm p=(B,F,\theta)$ and beamforming vector
$\bm w(\bm p)$
\STATE \textbf{Stage 1: Blockage-aware phase-matched target}
\FOR{$n=1,\ldots,N$}
	\STATE Evaluate $K_{z_r}(x_r,x_n)$ and $T_{\rm e}(x_n;x_r)$
	\STATE Set $\phi_{{\rm F},n}=-\arg K_{z_r}(x_r,x_n)$ and
	$\Delta\phi_{{\rm e},n}=-\arg T_{\rm e}(x_n;x_r)$
	\STATE Store $\psi_n^{(0)}=\phi_{{\rm F},n}+\Delta\phi_{{\rm e},n}$
	\STATE Set $w_n=A_G(x_n)|K_{z_r}(x_r,x_n)T_{\rm e}(x_n;x_r)|$
\ENDFOR
\STATE Unwrap $\{\psi_n^{(0)}\}_{n=1}^{N}$ continuously in increasing
$x_n$ to obtain $\{\psi_n\}_{n=1}^{N}$
\STATE \textbf{Stage 2: Weighted phase projection}
\STATE Accumulate the weighted $4\times4$ Gram matrix and phase moments
\STATE Solve \eqref{eq:discrete_projection} for
$(c_0,a_1,a_2,a_3)$
\STATE \textbf{Stage 3: Airy control conversion}
\IF{$(a_1,a_2,a_3)$ is outside the admissible coefficient intervals}
	\STATE Apply the finite active-set version of
	\eqref{eq:constrained_phase_projection} using the same discrete moments
\ENDIF
\STATE Convert $(a_1,a_2,a_3)$ to $(B,F,\theta)$ using
\eqref{eq:coefficient_control_inverse}
\STATE Construct $\bm w(\bm p)$ using
\eqref{eq:discrete_beamforming_vector}
\end{algorithmic}
\end{algorithm}

\subsection{Complexity Analysis}

The projection problem has fixed dimension and therefore introduces only
constant complexity. The remaining operations, including effective-channel
evaluation, weighted-moment accumulation, and sampled phase generation, scale
linearly with the number of antenna elements. The total complexity is therefore
\begin{equation}
	\mathcal O(N).
\end{equation}
The moments can also be accumulated sequentially, so the additional memory
can be constant apart from the output beamforming vector.
Consequently, the proposed framework provides a direct geometry-to-Airy
mapping with linear complexity in the number of antenna elements. It does not
search over candidate trajectories or Airy configurations; instead, it maps
the blockage-induced phase requirement directly onto the realizable Airy
phase basis and then to $(B,F,\theta)$.

\section{Numerical Results}
\label{sec:results}

This section validates the proposed phase-compensation mechanism through
numerical evaluations. We first quantify the gain contribution of successive
phase components, then visualize the field-level consequence of the identified
compensation. We next characterize when the Airy cubic correction becomes
relevant and when cubic-order compensation is sufficient, and finally verify
the accuracy of the analytical realization.

\subsection{Simulation Setup}

The carrier frequency is $140$ GHz. The transmitter employs a $256$-element
ULA with half-wavelength spacing, resulting in a $0.273$-m aperture. The
Gaussian aperture width is $\omega_0=a$, and the default receiver-window
width is $10$ mm. All methods use the same aperture amplitude and
transmit-power normalization.

The focused Gaussian beam is used as the conventional near-field baseline,
with $B=0$, $F=z_r$, and
$\theta=\arcsin(x_r/z_r)$. The feasible Airy-control range is set as
$|B|\leq3$, $0.25z_r\leq F\leq1.20z_r$, and
$|\theta|\leq5^\circ$. We evaluate $720$ single-edge blockage geometries
generated independently of any beam-design result. Continuous aperture and
receiver integrations are computed using Gauss--Legendre quadrature.

\begin{figure*}[!t]
    \centering
    \includegraphics[width=0.98\textwidth]{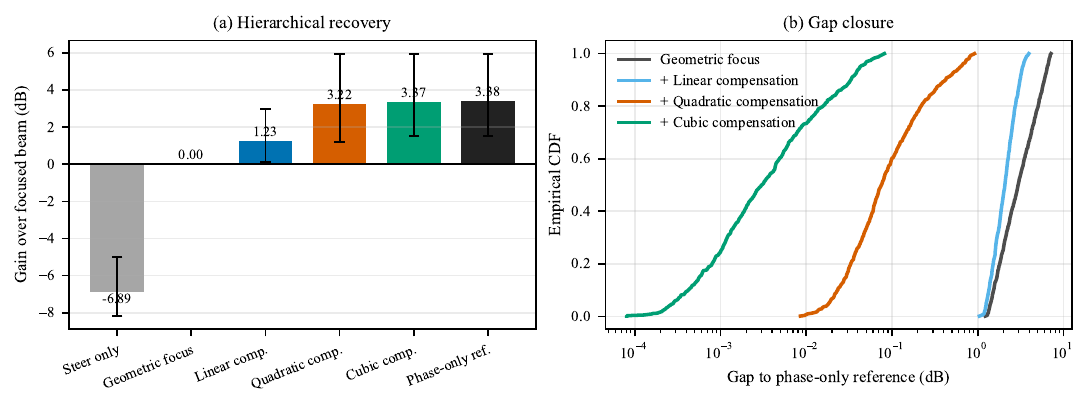}
    \caption{Effectiveness of phase compensation under edge blockage.
    (a) Received-power recovery from conventional focusing through successive
    linear, quadratic, and cubic phase compensation. (b) Empirical cumulative
    distribution of the gap to the phase-only upper bound across the evaluated
    blockage geometries.}
    \label{fig:phase_order_hierarchy}
\end{figure*}

\begin{figure*}[!t]
    \centering
    \includegraphics[width=0.91\textwidth]{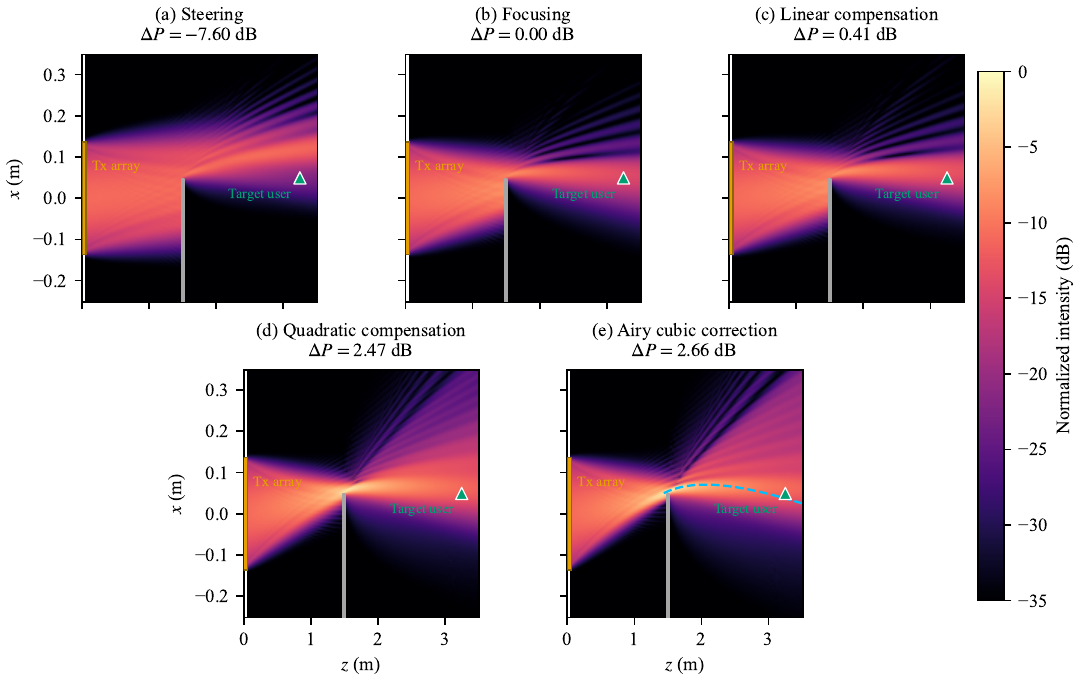}
    \caption{Propagated-field visualization of the phase-compensation
    mechanism. All cases use the same Gaussian aperture and transmit power;
    only the source phase is changed. The panel annotations $\Delta P$ denote
    received-power gain relative to conventional focusing. The yellow rectangle
    and green triangle mark the transmit array and target user, respectively.}
    \label{fig:full_space_comparison}
\end{figure*}

\begin{figure*}[!t]
    \centering
    \includegraphics[width=0.47\textwidth]{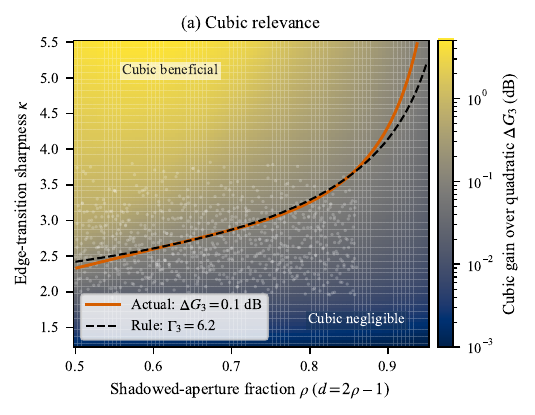}
    \hfill
    \includegraphics[width=0.47\textwidth]{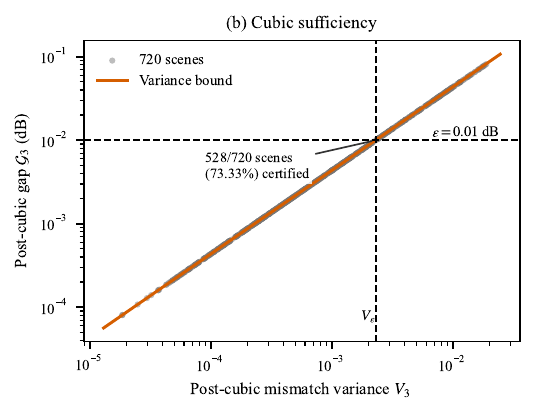}
    \caption{Cubic relevance and sufficiency. (a) Geometry dependence of the
    cubic gain $\Delta G_3$ over quadratic compensation. The solid contour marks
    $\Delta G_3=0.1$ dB, the dashed contour denotes the analytical
    $\Gamma_3=6.2$ relevance rule, and the dots represent the evaluated scenes.
    (b) Residual-variance certificate after cubic compensation. The solid curve
    is the bound in \eqref{eq:variance_gap_certificate}; at $0.01$ dB, the
    dashed thresholds certify $528$ scenes ($73.33\%$).}
    \label{fig:phase_order_geometry_boundary}
\end{figure*}

\subsection{Effectiveness of Phase Compensation}

We first evaluate whether the proposed phase-compensation framework can
effectively recover the received power lost to blockage. Fig.~\ref{fig:phase_order_hierarchy}
shows the gain contribution of successive phase components relative to the
phase-only upper bound; panel (b) reports the corresponding empirical
cumulative distribution function (CDF).

Starting from conventional near-field focusing, linear compensation improves
the received power by $1.228$ dB on average, while quadratic compensation
provides the dominant additional gain of $1.995$ dB. The Airy cubic component
contributes an additional $0.1433$ dB on average, thereby closing the residual
gap to the phase-only upper bound from $0.1532$ dB to $0.00991$ dB.
Equivalently, in linear power, the cubic stage reaches $99.77\%$ of the
phase-only upper bound.

These results demonstrate the central phase-compensation mechanism:
the quadratic phase degree of freedom provides the dominant blockage
compensation, while the Airy cubic component serves as a higher-order
refinement for the remaining phase mismatch. Therefore, the curved Airy trajectory
should be understood as a propagated consequence of source-phase compensation
rather than the fundamental origin of the recovery.

\subsection{Field-Level Interpretation of Phase Compensation}

We next visualize the field-level consequence of the identified phase
compensation in Fig.~\ref{fig:full_space_comparison}.
To connect the propagated field with its source-phase mechanism,
Fig.~\ref{fig:phase_mismatch_extraction}
compares the conventional focusing phase with the phase-matched reference. The
difference between them represents the additional phase requirement introduced
by blockage. Fig.~\ref{fig:hierarchical_phase_compensation} then shows its
successive compensation: linear compensation removes the first-order
variation, quadratic compensation removes the dominant curvature variation,
and the Airy cubic component corrects the remaining non-quadratic residual.

The field evolution confirms the same hierarchy observed in received-power
recovery: quadratic compensation provides the dominant improvement, while the
Airy cubic component provides the final refinement. The curved trajectory
appears only after source-phase compensation is applied and therefore
represents the spatial manifestation of the compensated wavefront, rather than
the mechanism that produces the recovery.

\subsection{Cubic Relevance and Sufficiency}

Fig.~\ref{fig:phase_order_geometry_boundary} jointly characterizes when the
Airy cubic is beneficial and when it is sufficient. In panel (a), the color
denotes the received-power increment $\Delta G_3$ obtained by adding the cubic
component after quadratic compensation. The $0.1$-dB contour separates
geometries in which the cubic correction is negligible from those in which it
is beneficial. Cubic relevance is jointly controlled by the edge-transition
sharpness and the clear-side aperture span, rather than by the blockage ratio
alone. Panel (b) provides the complementary post-cubic certificate: at the
$0.01$-dB tolerance, $V_3\leq V_{\varepsilon_{\rm dB}}$ certifies $528$ of the
$720$ evaluated scenes, or $73.33\%$.

\subsection{Accuracy of the Analytical Realization}

We finally verify whether the identified phase compensation can be accurately
realized within the Gaussian-cubic Airy control family. The proposed weighted
phase projection achieves a mean gap of $9.19\times10^{-5}$ dB to continuous
phase optimization over all 720 geometries. On the 72-scene stratified subset,
its mean and maximum gaps to independently initialized broad search are
$8.62\times10^{-5}$ dB and $2.93\times10^{-4}$ dB, respectively. Detailed
distributional and numerical consistency checks are provided in
Appendix~\ref{app:numerical_validation}.

\section{Conclusions}

This paper explained why cubic-phase Airy beamforming is sufficient for
blockage recovery from a phase-compensation perspective. By identifying the
additional phase mismatch relative to conventional near-field focusing, we
revealed that the existing linear and quadratic phase degrees of freedom can
be reoptimized under blockage to provide resteering and refocusing. The
quadratic term compensates the dominant quadratic component of the mismatch,
while the Airy cubic provides the first independent correction to the
remaining non-quadratic component. Numerical results validate this mechanism:
lower-order phase terms recover most of the available gain, while the cubic
term closes nearly all of the remaining gap and enables the Airy family to
reach $99.77\%$ of the phase-only upper bound.
More broadly, these results shift the interpretation of Airy blockage
recovery from trajectory design to phase compensation. The curved trajectory
is a propagated consequence of the compensated source phase rather than the
physical origin of received-power recovery. This viewpoint suggests a general
principle for structured beamforming under blockage: first characterize the
phase requirement imposed by the blocked channel, then exploit existing
lower-order phase degrees of freedom before introducing higher-order
structure. The sufficiency of a structured beam can therefore be assessed
through its remaining phase mismatch.
Future work will extend this phase-compensation analysis beyond the current
single-edge Fresnel model to more general blockage scenarios and practical
array implementations.

\appendices
\section{Local Edge Approximation and Geometry Reduction}
\label{app:geometry_reduction}

\subsection{Local Edge-Intersection Approximation}

The full weighted phase fit in \eqref{eq:weighted_phase_projection} is already
one-shot, but it still requires weighted moment accumulation. A local
approximation makes the dependence of the Airy controls on the edge geometry
explicit. Define the normalized aperture coordinate whose geometrical
source-to-receiver contribution intersects the obstacle edge:
\begin{equation}
	t_{\rm e}=\frac{z_rx_e-z_ox_r}{(z_r-z_o)a}.
	\label{eq:edge_center_coordinate}
\end{equation}
At $t=t_{\rm e}$, the argument of \eqref{eq:edge_transfer} is zero. The local rule
below applies when $t_{\rm e}\in[-1,1]$, so the edge transition intersects the
illuminated aperture. When this condition or the physical-branch conditions
are not satisfied, the full weighted phase fit in
\eqref{eq:weighted_phase_projection} is used instead. Define
\begin{equation}
	\kappa
	=a\sqrt{\frac{\pi(z_r-z_o)}{\lambda z_oz_r}},
	\qquad
	r_2=\frac{2}{\sqrt{\pi}}\left(\frac{4}{\pi}-1\right).
	\label{eq:edge_sharpness}
\end{equation}
For $\ell=1,2,3$, define
$\psi_\ell=\left.\mathrm d^\ell\psi_{\rm req}/\mathrm dt^\ell
\right|_{t=t_{\rm e}}$. The first three derivatives are
\begin{equation}
\begin{aligned}
	\psi_1&=\frac{2\pi a(x_r-at_{\rm e})}{\lambda z_r}
	+s\sqrt{\frac{2}{\pi}}\kappa,\\
	\psi_2&=-\frac{2\pi a^2}{\lambda z_r}
	-\frac{4}{\pi}\kappa^2,\\
	\psi_3&=s\sqrt{2}\,r_2\kappa^3.
	\label{eq:required_phase_derivatives}
\end{aligned}
\end{equation}
We call this the edge-intersection (EI) expansion. Matching these derivatives
with $a_1t+a_2t^2+a_3t^3$ gives
\begin{equation}
\begin{aligned}
	a_3^{\rm EI}&=\frac{\psi_3}{6},\\
	a_2^{\rm EI}&=\frac{\psi_2}{2}-3a_3^{\rm EI}t_{\rm e},\\
	a_1^{\rm EI}&=\psi_1-t_{\rm e}\psi_2+3a_3^{\rm EI}t_{\rm e}^2.
	\label{eq:explicit_coefficients}
\end{aligned}
\end{equation}
In particular, the physical bending scale is
\begin{equation}
	\boxed{
	\begin{aligned}
	B_{\rm EI}
	&=C_Bs\sqrt{\frac{z_r-z_o}{\lambda z_oz_r}},\\
	C_B&=\frac{1}{2\sqrt{\pi}}
	\left(\frac{r_2}{\sqrt{2}}\right)^{1/3}
	\approx0.1698.
	\end{aligned}
	}
	\label{eq:explicit_bending}
\end{equation}
The bending direction is determined by the visible side, while its magnitude
follows a Fresnel geometry scale. Equation~\eqref{eq:explicit_bending} is a
third-order local phase-matching rule. It is not claimed to be the exact
finite-aperture optimum, for which
\eqref{eq:weighted_phase_projection} is the more accurate analytical design.

\subsection{Two-Parameter Geometry Reduction}

The original geometry information $\mathcal I$ contains five variables, which
obscures when a cubic correction is actually needed. The preceding local
approximation suggests that the full weighted phase fit may depend on fewer
geometric combinations. To expose them, define the mirrored coordinate,
normalized edge position, and geometric blockage ratio as
\begin{equation}
	\tau=st,\qquad
	d=st_{\rm e}=2\rho-1,\qquad
	\rho=\frac{1+st_{\rm e}}{2}.
	\label{eq:dimensionless_geometry}
\end{equation}
Here, mirroring by $s$ places both visible-side cases in one coordinate
system, $d$ locates the edge transition within the normalized aperture, and
$\rho$ is the geometrically shadowed aperture fraction. The parameter
$\kappa$ in \eqref{eq:edge_sharpness} is the edge-transition sharpness: it
measures how rapidly the edge-diffraction response varies across the aperture.
Substituting \eqref{eq:edge_center_coordinate} and
\eqref{eq:edge_sharpness} into \eqref{eq:edge_transfer} gives the exact
normalized edge-diffraction response
\begin{equation}
	T_{\rm e}(\tau;\kappa,d)
	=\frac{1}{2}\operatorname{erfc}\left[
	e^{-\jmathunit\pi/4}\kappa(d-\tau)\right].
	\label{eq:dimensionless_edge_transfer}
\end{equation}

\begin{proposition}[Two-parameter geometry reduction and stage sufficiency]
\label{prop:geometry_order_certificate}
For the fixed even Gaussian aperture and a point receiver, every normalized
power or phase-mismatch quantity produced by a stage-$m$ weighted phase fit,
including $V_m$, $P_m^{\rm proj}/P_{\rm ub}$, and the cubic capture ratios,
depends on $\mathcal I$ only through $(\kappa,d)$ and not on the visible-side
sign separately. Moreover, if
\begin{equation}
	\mathcal G_m
	=10\log_{10}\frac{P_{\rm ub}}{P_m^{\rm proj}},
	\label{eq:projected_gap_definition}
\end{equation}
then, for $V_m<2$,
\begin{equation}
	\boxed{
	\mathcal G_m
	\leq-20\log_{10}\left(1-\frac{V_m}{2}\right).
	}
	\label{eq:variance_gap_certificate}
\end{equation}
Consequently, for any prescribed tolerance $\varepsilon_{\rm dB}>0$,
$\mathcal G_m\leq\varepsilon_{\rm dB}$ dB is guaranteed whenever
\begin{equation}
	V_m\leq V_{\varepsilon_{\rm dB}}
	=2\left(1-10^{-\varepsilon_{\rm dB}/20}\right).
	\label{eq:variance_tolerance}
\end{equation}
\end{proposition}

\begin{IEEEproof}
The reduction follows directly from the normalized edge response
$T_{\rm e}(\tau;\kappa,d)$: the even Gaussian envelope and the polynomial
phase spaces remove any separate dependence on the visible-side sign. The
sufficiency bound follows from
$|\E_W[e^{-\jmathunit e_m}]|\geq1-V_m/2$; taking decibels and solving for the
prescribed tolerance give \eqref{eq:variance_gap_certificate} and
\eqref{eq:variance_tolerance}, respectively.
\end{IEEEproof}

Fig.~\ref{fig:phase_order_geometry_boundary}(b) verifies the certificate for
the cubic stage. At $\varepsilon_{\rm dB}=0.01$ dB, the condition
$V_3\leq V_{\varepsilon_{\rm dB}}$ certifies $528$ of the $720$ evaluated
scenes, or $73.33\%$.

The remaining phase-mismatch variance in
\eqref{eq:variance_tolerance} is a strict certificate that beamforming stage
$m$ is sufficient, but it is available only after the weighted phase fit has
been evaluated. For a simpler preliminary indication of cubic relevance,
\eqref{eq:required_phase_derivatives} gives the edge-induced third derivative
at the transition. A quadratic phase supplies constant curvature, whereas the
Airy cubic component supplies its first spatial variation. This motivates the
cubic-relevance indicator
\begin{equation}
	\boxed{
	\Gamma_3
	=|\psi_3|(1-d)
	=\sqrt{2}\,r_2\kappa^3(1-d),
	}
	\label{eq:curvature_variation_number}
\end{equation}
where $1-d$ is the normalized clear-side width from the edge to the aperture
endpoint. Hence cubic relevance is controlled jointly by edge-transition
sharpness and the aperture span over which the induced curvature varies,
rather than by the blockage ratio alone. Section~\ref{sec:results} tests
\eqref{eq:curvature_variation_number} as a low-cost cubic-relevance rule; the
strict beamforming-stage sufficiency condition remains
\eqref{eq:variance_tolerance}.

\subsection{Physical Interpretation of Cubic Relevance}

Define the effective edge distance and edge Fresnel number
\begin{equation}
	L_{\rm e}=\frac{z_oz_r}{z_r-z_o},\qquad
	\mathcal F_{\rm e}
	=\frac{a^2}{\lambda L_{\rm e}}
	=\frac{a^2(z_r-z_o)}{\lambda z_oz_r}.
	\label{eq:edge_fresnel_number}
\end{equation}
Then $\kappa^2=\pi\mathcal F_{\rm e}$, and
\begin{equation}
	\boxed{
	\begin{aligned}
	\Gamma_3
	&=C_\Gamma(1-d)\mathcal F_{\rm e}^{3/2}\\
	&=2C_\Gamma(1-\rho)\mathcal F_{\rm e}^{3/2},\\
	C_\Gamma
	&=\sqrt{2}r_2\pi^{3/2}\approx2.4279.
	\end{aligned}
	}
	\label{eq:gamma_fresnel_scaling}
\end{equation}
Equation~\eqref{eq:gamma_fresnel_scaling} provides a physical interpretation
of the cubic-relevance indicator in terms of the edge Fresnel number and
aperture geometry; the strict stage-sufficiency certificate remains
\eqref{eq:variance_tolerance}.

\subsection{Weak-Blockage Limit}

\begin{corollary}[Free-space degeneration]
\label{cor:free_space}
If the edge-diffraction response approaches a constant complex gain across the
illuminated aperture, then no independent Airy cubic correction remains. Hence
\begin{equation}
	B\rightarrow0,\qquad F\rightarrow z_r,\qquad
	\sin\theta\rightarrow x_r/z_r.
	\label{eq:free_space_degeneration}
\end{equation}
\end{corollary}

Corollary~\ref{cor:free_space} establishes the weak-blockage consistency of the
phase-compensation interpretation. As the edge ceases to distort the aperture
phase, the structured cubic realization returns continuously to focused
Gaussian transmission without a manually selected blockage threshold.

\section{Theoretical Proofs and Realization Details}
\label{app:theoretical_details}

This appendix provides the formal proofs and then bridges the continuous
phase-compensation theory to discrete-array and finite-window realizations.

\subsection{Proofs}

\begin{IEEEproof}[Proof of Theorem~\ref{thm:phase_projection}]
After removal of the weighted mismatch mean, the normalized field is
$\E_W[e^{\jmathunit\delta}]$. Taylor expansion gives
\begin{equation}
	\E_W[e^{\jmathunit\delta}]
	=1-\frac{1}{2}\E_W[\delta^2]+\mathcal R_{\delta,3},
	\qquad |\mathcal R_{\delta,3}|
	\leq\frac{1}{6}\E_W|\delta|^3.
\end{equation}
Because $\E_W[\delta]=0$, the first-order term vanishes. Taking the squared
magnitude of the expansion above gives
\eqref{eq:coherent_power_expansion}; the assumed uniform bound on $\delta$
also absorbs the fourth-order product into the stated third-order remainder.

It remains to determine the phase coefficients that minimize the second-order
loss. Let $r_{\bm q}(t)=\bm q^{\mathsf T}\bm X(t)-\psi_{\rm req}(t)$. Since the
constant basis function is included in $\bm X$, optimizing $c_0$ centers
$r_{\bm q}$, and minimizing $\E_W[\delta^2]$ is therefore equivalent to
minimizing the weighted least-squares objective
\begin{equation}
	J(\bm q)=\int_{-1}^{1}W(at)r_{\bm q}^2(t)\,\mathrm dt.
\end{equation}
Its stationarity condition is
\begin{equation}
	\left(\int_{-1}^{1}W(at)\bm X\bm X^{\mathsf T}\,\mathrm dt\right)\bm q
	=\int_{-1}^{1}W(at)\bm X\psi_{\rm req}\,\mathrm dt.
\end{equation}
When the weighted Gram matrix is nonsingular, solving this normal equation
gives \eqref{eq:weighted_phase_projection}, completing the proof.
\end{IEEEproof}

\begin{IEEEproof}[Proof of Proposition~\ref{prop:cubic_capture}]
Since $r_3$ is orthogonal to $\mathcal V_2$, extending
$\mathcal V_2$ to $\mathcal V_3$ adds exactly one orthogonal direction.
Moreover, $e_2=(I-\Pi_2^W)\psi_{\rm req}$ implies
$\langle\psi_{\rm req},r_3\rangle_W=\langle e_2,r_3\rangle_W$. Hence
orthogonal projection onto $\mathcal V_3$ gives
\begin{equation}
	e_3=e_2-
	\frac{\langle e_2,r_3\rangle_W}{\langle r_3,r_3\rangle_W}r_3.
\end{equation}
The two terms on the right-hand side are orthogonal. The Pythagorean identity
therefore gives \eqref{eq:cubic_residual_reduction}, and division by
$\|e_2\|_W^2$ gives \eqref{eq:cubic_capture_factor}.
\end{IEEEproof}

\subsection{Discrete Airy Realization}

The projected cubic coefficient in \eqref{eq:cubic_edge_projection} maps to
the physical bending parameter as
\begin{equation}
	B^{\rm proj}
	=\frac{\sqrt[3]{3a_3^{\rm proj}}}{2\pi a}.
	\label{eq:projected_bending}
\end{equation}

\begin{corollary}[Discrete-array phase fitting]
\label{cor:discrete_projection}
For the physical ULA, let $A_n=A_G(x_n)$ and
$h_n=h_{\rm e}(x_n;x_r)$. Define $t_n=x_n/a$, let
$\bm X_N\in\mathbb R^{N\times4}$ have row
$[1,t_n,t_n^2,t_n^3]$, set
$\bm W_N=\operatorname{diag}(A_n|h_n|)$, and define
$[\bm\psi_{\rm req}]_n$ as the continuous unwrapped samples of
$-\arg h_n$. Replacing the continuous moments by sums gives
\begin{equation}
	\widehat{\bm q}_N
	=(\bm X_N^{\mathsf T}\bm W_N\bm X_N)^{-1}
	\bm X_N^{\mathsf T}\bm W_N\bm\psi_{\rm req}.
	\label{eq:discrete_projection}
\end{equation}
Discretization therefore changes the weighted moments, not the
phase-compensation mechanism.
\end{corollary}

The unconstrained fit may correspond to Airy controls outside their physical
ranges. Define the coefficient-domain image of the feasible controls as
\begin{equation}
	\mathcal Q_{\rm A}
	=\{[c_0,a_1,a_2,a_3]^{\mathsf T}:
	c_0\in\R,\ (B,F,\theta)\in\mathcal P_{\rm A}\}.
\end{equation}
The corresponding constrained projection is
\begin{equation}
	\widehat{\bm q}_{\rm A}
	=\arg\min_{\bm q\in\mathcal Q_{\rm A}}
	\int_{-1}^{1} W(at)\left[\bm q^{\mathsf T}\bm X(t)
	-\psi_{\rm req}(t)\right]^2\,\mathrm dt.
	\label{eq:constrained_phase_projection}
\end{equation}
For box-bounded physical controls, this fixed-dimensional convex problem is
solved by finite active-set enumeration and reduces to
\eqref{eq:weighted_phase_projection} for an interior solution.

\subsection{Finite-Window Extension}

Define $P(\bm a,x)=|E_{\rm b}(x;\bm a)|^2$ and the average power over a
receiver window of width $W_r$ as
$\overline P_{W_r}(\bm a)=W_r^{-1}\int_{x_r-W_r/2}^{x_r+W_r/2}
P(\bm a,x)\,\mathrm dx$.

\begin{proposition}[Finite receiver-window perturbation]
\label{prop:finite_window}
Suppose that $P$ is four-times continuously differentiable near
$(\bm a_0,x_r)$ and that $\bm a_0$ is a strict point-receiver maximizer with
nonsingular Hessian
$\bm H_0=\nabla_{\bm a}^{2}P(\bm a_0,x_r)$. The nearby finite-window
maximizer $\bm a_{W_r}$ then satisfies
\begin{equation}
	\bm a_{W_r}
	=\bm a_0-\frac{W_r^2}{24}\bm H_0^{-1}
	\nabla_{\bm a}\partial_x^2P(\bm a_0,x_r)
	+\mathcal O(W_r^4),
	\label{eq:window_parameter_perturbation}
\end{equation}
and the window-power loss from using $\bm a_0$ instead of
$\bm a_{W_r}$ is $\mathcal O(W_r^4)$.
\end{proposition}

\begin{IEEEproof}
A symmetric Taylor expansion about $x_r$ gives
$\overline P_{W_r}(\bm a)=P(\bm a,x_r)
+W_r^2\partial_x^2P(\bm a,x_r)/24+\mathcal O(W_r^4)$.
Expanding its stationarity condition about $\bm a_0$ yields
\eqref{eq:window_parameter_perturbation}. Because
$\bm a_{W_r}-\bm a_0=\mathcal O(W_r^2)$, the associated loss in the
finite-window objective is second order in this displacement and hence
$\mathcal O(W_r^4)$.
\end{IEEEproof}

This result connects the point-receiver phase analysis used for the analytical
derivation to the finite receiver window used in the simulations. Its
numerical accuracy is assessed in Appendix~\ref{app:numerical_consistency}.

\section{Numerical Validation and Scope Checks}
\label{app:numerical_validation}

This appendix validates the numerical references, identifies the propagation
regime in which the analytical model is used, and checks the consistency of
the reduced channel and practical realization.

\subsection{Reference Optimization Validation}

The stage references were audited by deterministic scans and independently
initialized searches over the adopted coefficient domain. The largest audit
advantage over the reported references is below $10^{-14}$ dB. Adding the
quartic phase term beyond the Airy cubic provides only $0.00815$ dB mean gain,
supporting cubic near-saturation without claiming exact higher-order
sufficiency.

An independently initialized broad search over 72 stratified scenes and an
expanded control domain has mean and maximum advantages of
$8.62\times10^{-5}$ dB and $2.93\times10^{-4}$ dB, respectively, over the
one-shot phase fit. All analytical solutions remain inside the narrower
reference domain, so the bounded fallback is not activated in the reported
experiments.

Fig.~\ref{fig:projection_validation} reports the corresponding empirical gap
distributions for the continuous local reference and the independently
initialized broad search.

\begin{figure}[!t]
    \centering
    \includegraphics[width=0.96\columnwidth]{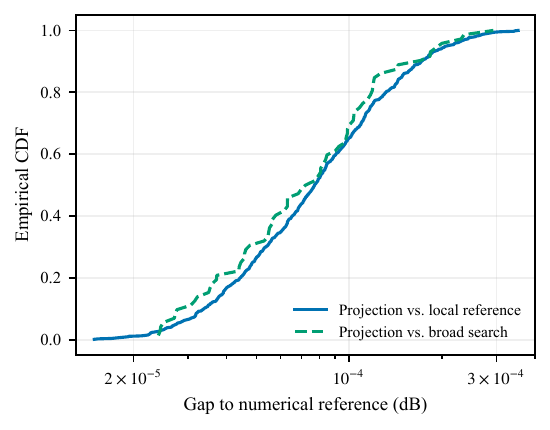}
    \caption{Accuracy of the analytical phase realization. Empirical cumulative
    distributions compare the weighted projection with the continuous local
    reference over all $720$ geometries and the independently initialized broad
    search over the $72$-scene stratified subset.}
    \label{fig:projection_validation}
\end{figure}

\subsection{Propagation Model Validation}

The Fresnel range expansion is audited through its omitted phase
\begin{equation}
	\epsilon_{\rm F}(v,\ell)
	=k\ell\left(1+\frac{v^2}{2}-\sqrt{1+v^2}\right),
	\label{eq:fresnel_phase_remainder}
\end{equation}
where $v$ is a transverse-to-axial slope over distance $\ell$. A scene
satisfies the adopted Fresnel criterion when the maximum remainder over the
source, Airy-propagation, transmitter-to-edge, and edge-to-receiver segments
does not exceed $0.5$ rad. The fitted controls satisfy this criterion in
$92.6\%$ of the scenes; the remaining scenes are not used to claim validity
beyond the paraxial regime.

As an independent propagation check, the same Gaussian amplitudes and phase
coefficients were evaluated by a scalar angular-spectrum method (ASM), without
reoptimizing any method, over 72 stratified Fresnel-valid scenes. The mean
phase-fit gains over focus are $3.2320$ dB and $3.2404$ dB under the analytical
Fresnel and ASM models, respectively, with a mean absolute difference of
$0.0126$ dB. Under ASM, the linear, quadratic, and cubic gain increments are
$1.1690$, $1.9276$, and $0.1439$ dB, respectively. The independent propagator
therefore preserves the lower-order-dominant hierarchy and the smaller,
conditional cubic refinement.

\subsection{Numerical Consistency Checks}
\label{app:numerical_consistency}

\noindent\textbf{1) Effective-channel reduction and quadrature:}
The single-integral effective-channel reduction and direct two-step Fresnel
propagation have a mean rank correlation of $0.99893$ over 18 scenes and 273
Airy configurations. Increasing the aperture and receiver quadrature orders
changes the fitted coefficients negligibly, and the continuously unwrapped
phase branch is numerically stable.

\noindent\textbf{2) Finite-window and discrete-array effects:}
The finite-window and discrete-array limits were evaluated on 72 scenes
selected evenly across blockage, obstacle distance, and visible side. For the
default 10-mm receiver window, the center-based phase fit is only
$8.62\times10^{-5}$ dB below the corresponding window-power reference. When
the window width is doubled to 20 mm, the gap increases by approximately a
factor of 16, in agreement with the $\mathcal O(W_r^4)$ prediction in
\eqref{eq:window_parameter_perturbation}. Even at 40 mm, the mean gap remains
$0.0196$ dB.

For the half-wavelength 256-element ULA, applying the continuous analytical
controls instead of recomputing the sampled-aperture phase fit causes only
$1.97\times10^{-5}$ dB mean loss. The largest adjacent phase step is
$0.571$ rad, below $\pi$, across all tested scenes. Sampling therefore changes
the numerical moments without changing the underlying compensation mechanism.
Table~\ref{tab:consistency_checks} consolidates the principal checks across
this appendix.

\begin{table}[!t]
\centering
\caption{Analytical and Numerical Consistency Checks}
\label{tab:consistency_checks}
\setlength{\tabcolsep}{3.5pt}
\begin{tabular}{lcc}
\toprule
Check & Mean gap & Maximum gap \\
\midrule
Stage audit (worst stage)
& $1.54{\times}10^{-15}$ dB & $7.71{\times}10^{-15}$ dB \\
Phase fit to 10-mm reference (720)
& $9.19{\times}10^{-5}$ dB & $3.53{\times}10^{-4}$ dB \\
Phase fit to broad search (72)
& $8.62{\times}10^{-5}$ dB & $2.93{\times}10^{-4}$ dB \\
256-element discretization
& $1.97{\times}10^{-5}$ dB & $8.61{\times}10^{-5}$ dB \\
\bottomrule
\end{tabular}
\end{table}

\bibliographystyle{IEEEtran}
\bibliography{ref}

\begin{thebibliography}{10}
\providecommand{\url}[1]{#1}
\csname url@samestyle\endcsname
\providecommand{\newblock}{\relax}
\providecommand{\bibinfo}[2]{#2}
\providecommand{\BIBentrySTDinterwordspacing}{\spaceskip=0pt\relax}
\providecommand{\BIBentryALTinterwordstretchfactor}{4}
\providecommand{\BIBentryALTinterwordspacing}{\spaceskip=\fontdimen2\font plus
\BIBentryALTinterwordstretchfactor\fontdimen3\font minus \fontdimen4\font\relax}
\providecommand{\BIBforeignlanguage}[2]{{%
\expandafter\ifx\csname l@#1\endcsname\relax
\typeout{** WARNING: IEEEtran.bst: No hyphenation pattern has been}%
\typeout{** loaded for the language `#1'. Using the pattern for}%
\typeout{** the default language instead.}%
\else
\language=\csname l@#1\endcsname
\fi
#2}}
\providecommand{\BIBdecl}{\relax}
\BIBdecl

\bibitem{2021THz}
T.~K{\"u}rner, D.~M. Mittleman, and T.~Nagatsuma, Eds., \emph{{THz} Communications: Paving the Way Towards Wireless {Tbps} (Springer Series in Optical Sciences, 234)}.\hskip 1em plus 0.5em minus 0.4em\relax Springer Cham, 2021.

\bibitem{8732419}
T.~S. Rappaport, Y.~Xing, O.~Kanhere, S.~Ju, A.~Madanayake, S.~Mandal, A.~Alkhateeb, and G.~C. Trichopoulos, ``Wireless communications and applications above 100 {GHz}: Opportunities and challenges for {6G} and beyond,'' \emph{IEEE Access}, vol.~7, pp. 78\,729--78\,757, 2019.

\bibitem{shurakov2023empirical}
A.~Shurakov, D.~Moltchanov, A.~Prikhodko, A.~Khakimov, E.~Mokrov, V.~Begishev, I.~Belikov, Y.~Koucheryavy, and G.~Gol'tsman, ``Empirical blockage characterization and detection in indoor {Sub-THz} communications,'' \emph{Comput. Commun.}, vol. 201, pp. 48--58, 2023.

\bibitem{overview}
N.~K. Efremidis, Z.~Chen, M.~Segev, and D.~N. Christodoulides, ``{Airy} beams and accelerating waves: an overview of recent advances,'' \emph{Optica}, vol.~6, no.~5, pp. 686--701, May 2019.

\bibitem{airyfinite}
G.~A. Siviloglou and D.~N. Christodoulides, ``Accelerating finite energy {Airy} beams,'' \emph{Opt. Lett.}, vol.~32, no.~8, pp. 979--981, Apr. 2007.

\bibitem{airy1}
\BIBentryALTinterwordspacing
M.~V. Berry and N.~L. Balazs, ``Nonspreading wave packets,'' \emph{Am. J. Phys.}, vol.~47, no.~3, pp. 264--267, Mar. 1979. [Online]. Available: \url{https://doi.org/10.1119/1.11855}
\BIBentrySTDinterwordspacing

\bibitem{airy5}
G.~A. Siviloglou, J.~Broky, A.~Dogariu, and D.~Christodoulides, ``Observation of accelerating {Airy} beams,'' \emph{Phys. Rev. Lett.}, vol.~99, no.~21, p. 213901, Nov. 2007.

\bibitem{selfhealing1}
J.~Broky, G.~A. Siviloglou, A.~Dogariu, and D.~N. Christodoulides, ``Self-healing properties of optical {Airy} beams,'' \emph{Opt. Express}, vol.~16, no.~17, pp. 12\,880--12\,891, 2008.

\bibitem{selfhealing2}
X.~Chu, G.~Zhou, and R.~Chen, ``Analytical study of the self-healing property of {Airy} beams,'' \emph{Phys. Rev. A}, vol.~85, no.~1, p. 013815, 2012.

\bibitem{hanchong}
W.~Zhao, S.~Abadal, G.~Song, J.~Jiang, and C.~Han, ``{Terahertz} wireless data center: {Gaussian} beam or {Airy} beam?'' \emph{IEEE Trans. Wireless Commun.}, vol.~25, pp. 7922--7938, 2026.

\bibitem{songlingyang}
S.~Zhang, B.~Di, and L.~Song, ``Breaking near-field communication barriers: Focused, curved, or {Airy} beamforming?'' \emph{arXiv preprint arXiv:2604.01704}, 2026.

\bibitem{wang2026multiairy}
Y.~Wang and L.~Dai, ``Blockage-robust beamforming for near-field communications: From single-airy to multi-airy,'' \emph{arXiv preprint arXiv:2607.07278}, 2026.

\bibitem{hanchong2}
W.~Zhao, C.~Han, and E.~Bj{\"o}rnson, ``{Airy} beam engineering in near-field communications: A tractable closed-form analysis in the {Terahertz} band,'' \emph{arXiv preprint arXiv:2603.13866}, 2026.

\bibitem{lee}
D.~Lee, Y.~Yagi, K.~Suzuoki, and R.~Kudo, ``Experimental demonstration of wireless transmission using {Airy} beams in {Sub-THz} band,'' \emph{IEEE Open J. Commun. Soc.}, vol.~6, pp. 1091--1102, Jan. 2025.

\bibitem{curving}
H.~Guerboukha, B.~Zhao, Z.~Fang, E.~Knightly, and D.~M. Mittleman, ``Curving {THz} wireless data links around obstacles,'' \emph{Commun. Eng.}, vol.~3, no.~1, p.~58, 2024.

\bibitem{zhao2026efficienttraining}
W.~Zhao and C.~Han, ``Efficient {Airy} beam training for quasi-{LoS} {Terahertz} near-field communications,'' \emph{arXiv preprint arXiv:2605.09895}, 2026.

\bibitem{wang2026airy}
\BIBentryALTinterwordspacing
Y.~Wang and L.~Dai, ``Physics-guided neural airy beamforming for near-field blockage mitigation,'' 2026. [Online]. Available: \url{https://arxiv.org/abs/2608.04388}
\BIBentrySTDinterwordspacing

\bibitem{optics}
J.~W. Goodman, \emph{Introduction to {Fourier} Optics}.\hskip 1em plus 0.5em minus 0.4em\relax McGraw-Hill, 1968.

\end{thebibliography}

\end{document}